%% file: main.tex
\documentclass[letterpaper, 10 pt, conference]{ieeeconf}  

\IEEEoverridecommandlockouts                              

\usepackage{todonotes}

\usepackage{xcolor}
\usepackage{psfrag}
\usepackage{placeins}

\usepackage{graphics} 
\usepackage{epsfig} 
\usepackage{times} 
\usepackage{amsmath} 
\usepackage{amssymb}  
\usepackage{mathtools}
\usepackage{array}
\usepackage{bm}
\usepackage{tikz}
\usetikzlibrary{calc,arrows}
\usepackage{pgfplots}
\usepackage{subcaption}
\usepackage{etex}
\usepackage{siunitx}
\usepackage{url}
\usepackage{arydshln} 
\usepackage[ruled,vlined,linesnumbered]{algorithm2e}
 \let\labelindent\relax
\usepackage{enumitem} 

\usepackage[noend]{algpseudocode}

\makeatletter
\def\BState{\State\hskip-\ALG@thistlm}
\makeatother

\DeclareMathOperator*{\diagonal}{diag}

\title{Distributed Learning Model Predictive Control for Linear Systems*}

\author{Yvonne R.\ St\"urz$^{1}$, Edward L.\ Zhu$^{1}$, Ugo Rosolia$^{2}$, Karl H.\ Johansson$^{3}$, Francesco Borrelli$^{1}$ 
\thanks{$^{1}$The authors are with the Model Predictive Control Laboratory, University of California, Berkeley, CA 94709, USA. Email address: %
        {\tt\small y.stuerz@berkeley.edu }} 
\thanks{$^{2}$The author is with the Department of Mechanical and Civil Engineering,
California Institute of Technology, Pasadena, CA 91125 USA}
\thanks{$^{3}$The author is with the 
School of Electrical Engineering and Computer Science, KTH Royal Institute of Technology, Stockholm, Sweden}
\thanks{*This project has received funding from the European Union’s Horizon 2020 research and innovation programme under the Marie Sklodowska-Curie grant agreement No.\ 846421.}}%

\include{nomenclature}

\begin{document}

\maketitle
\thispagestyle{empty}
\pagestyle{empty}

\begin{abstract}
	This paper presents a distributed learning model predictive control (DLMPC) scheme for distributed linear time invariant systems with coupled dynamics and state constraints.	
	The proposed solution method is based on an online  distributed optimization scheme with nearest-neighbor communication. 
	If the control task is iterative and data from previous feasible iterations are available, local data are exploited by the subsystems in order to construct the local terminal set and terminal cost, which guarantee recursive feasibility and asymptotic stability, as well as performance improvement over iterations. 
	In case a first feasible trajectory is difficult to obtain, or the task is  non-iterative, we further propose an algorithm that efficiently explores the state-space and generates the data required for the construction of the terminal cost and terminal constraint in the MPC problem in a safe and distributed way. 
	In contrast to other distributed MPC schemes which use structured positive invariant sets, the proposed approach involves a control invariant set as the terminal set, on which we do not impose any distributed structure. The proposed iterative scheme converges to the global optimal solution of the underlying infinite horizon optimal control problem under mild conditions. 
	Numerical experiments demonstrate the effectiveness of the proposed DLMPC scheme. 
\end{abstract}

\section{Introduction}

Complex systems composed of multiple subsystems are present in many control applications. 
The large scale and spatial distribution of these systems often make the control by a centralized unit intractable due to limitations in computation and communication. 
Research has therefore focused on proposing design schemes for local controllers which compute control actions for the individual subsystems based on only local information in decentralized schemes, and on communicated information from neighboring subsystems in distributed control schemes. 
One line of research has focused on exploiting the interconnection structure of the system in order to design interconnected controllers based on a convex reformulation involving linear matrix inequalities in a scalable way \cite{Stuerz2018a}. 
If constraints need to be accounted for, distributed model predictive control (DMPC) techniques can be employed. 
They can mainly be categorized into non-cooperative, such as tube-based  \cite{Riverso2011a}, and cooperative schemes \cite{Camponogara2002, Giselsson2014, Conte2016, Darivianakis2019}.
The latter often involve distributed optimization techniques  \cite{Bertsekas1989} where the subsystems communicate local information and agree on a solution, thus solving the optimization problem  cooperatively.

The main challenge in DMPC schemes is to enable distributed computation by decomposing the optimization problem into subproblems for the individual subsystems. 
Most of the DMPC approaches in the literature therefore impose the distributed structure of the system on the terminal set and cost function of the MPC problem \cite{Riverso2011a, Giselsson2014, Conte2016, Farokhi2014, Trodden2017, Lucia2015}. In particular, they first design a structured terminal controller and cost based on Lyapunov stability and then design structured positive invariant sets under this terminal controller, satisfying the constraints. 
Two aspects in these schemes can lead to conservatism: (1) Imposing structure on the terminal controllers and terminal sets, and (2) computing positive invariant sets for one specific choice of terminal controller which is fixed in the design phase, lead to a possibly small inner approximation of the maximal control invariant set. 
In order to mitigate the conservatism introduced by the imposed structure, some works have proposed to adapt the terminal sets based on the states of the subsystems in operation \cite{Trodden2017, Lucia2015}, 
\cite{Conte2016, Simon2014}. 
In \cite{Darivianakis2019}, the stabilizing terminal controller is also computed online within the MPC problem. 

In \cite{Rosolia2017}, a data-driven MPC scheme, Learning MPC (LMPC), was introduced, where previously seen data are exploited in order to construct the terminal components of the MPC problem. In \cite{Rosolia2019} this framework was extended to uncertain systems, and it was shown how the LMPC scheme can be used to iteratively enlarge the domain of the policy. 

In this paper, we propose a distributed LMPC (DLMPC) scheme, which is a significant extension to \cite{Rosolia2017}. 
The contributions of the paper are the following: 
\begin{itemize}
	\item We present a novel DLMPC scheme for linear systems able to handle coupled dynamics, coupled state constraints and coupled cost functions. 
	The main improvement w.r.t.\ existing DMPC approaches is fully distributed computations without imposing any structure on the terminal cost function or constraint set. 
This is achieved by exploiting previously seen local data by the individual subsystems in order to build local data driven terminal sets and terminal cost functions. A consensus on specific parameters in the construction of the local costs and constraints is achieved by distributed optimization which guarantees that the local terminal sets are a control invariant set and the sum of the local terminal cost functions is a Lyapunov function for the global system. 
This can considerably reduce conservatism w.r.t.\ DMPC  schemes that rely on finding a positive invariant terminal set under a fixed structured stabilizing terminal controller. 
	\item For iterative control tasks, given a first feasible trajectory, 
	the proposed scheme provides recursive feasibility and asymptotic stability. 
	Furthermore, we prove that the proposed DLMPC has a non-increasing control performance over  iterations and, under mild conditions,  converges to the global centralized optimal solution. 
	\item For non-iterative control tasks, or if an initial feasible trajectory is difficult to obtain, 
	we further present an algorithm that by iteratively performing the proposed DLMPC scheme with changing starting conditions leads to an enlargement of the domain of the DLMPC policy. 
	This can be used to safely explore the state space and to generate the required data in a sample efficient and distributed way. 
\end{itemize}

The paper is structured as follows. 
Section~\ref{sec:problemformulation} introduces the model of the distributed systems and the control task considered. 
Section~\ref{sec:LMPC} provides a brief review of the LMPC in \cite{Rosolia2017a}.  A decomposed formulation of the LMPC according to the distributed system structure is presented in  Section~\ref{sec:DLMPCformulation}. The fully distributed solution is proposed in Section~\ref{sec:ADMM}. Section~\ref{sec:numerics} provides numerical experiments before Section~\ref{sec:con} concludes the paper.

\textbf{Notation} 
Let $\mathbb{R}$ denote the set of real numbers. $\mathbb{N}$ and $\mathbb{N}_+$ denote the set of non-negative and positive natural numbers. 
We denote the transpose of a vector $v \in \mathbb{R}^{n}$ as $v^\top$, and its Euclidean norm as $\| v \|$. 
The matrix $M = \diag{M_1, ..., M_m}$ is the block-diagonal matrix with submatrices $M_i$ on its diagonal. 
The symbol $\succcurlyeq$ is used to indicate elementwise inequality. 
The identitiy matrix of dimension $n$ is denoted as $I_n$ and the vector of all ones is denoted as $\boldsymbol{1}$. 

\section{Problem Formulation}
\label{sec:problemformulation}

In this section, we present the model of the distributed systems considered in this paper, and then state the control problem formulation. 

\subsection{Dynamically Coupled Constrained Linear Systems}
We consider the discrete-time linear time-invariant system with dynamics given by 
\begin{equation}\label{eq:sys_glob}
x_{t+1} = A x_t + B u_t,
\end{equation}
where $x_t \in \mathbb{R}^n$ and $u_t \in \mathbb{R}^m$ are the system state and input at time $t \in \mathbb{N}$. The system matrices $A$ and $B$ are assumed to be known. The system states and inputs are subject to linear constraints 
\begin{equation}
\begin{aligned}
x_t \in \mathcal{X}, \qquad \qquad u_t \in \mathcal{U}, 
\end{aligned}
\end{equation}
which are formulated as 
\begin{equation}
\begin{aligned}
G x_t \leq g, \qquad \qquad L u_t \leq l, 
\end{aligned}
\end{equation}
with $G$, $L$, $g$, and $l$ given matrices and vectors, respectively.

We consider systems in \eqref{eq:sys_glob} which have a structure that admits a decomposition into  subsystems $\N = \{1, ... , \numAgent\}$ which may be coupled in their state dynamics. 
The state of the $i$th subsystem is $x_{i,t} \in \mathbb{R}^{n_i}$,
and we assume that the $i$th input $u_{i,t} \in \mathbb{R}^{m_i}$ affects only the $i$-th state. 
Thus, the system states and inputs are partitioned as 
\begin{equation}\label{eq:xu_loc}
x_t = \bma{@{}c@{\,\,}c@{\,\,\,\,}c@{}}{x_{1,t}^{\top} & \hdots &  x_{M,t}^{\top}}^\top, \quad 
u_t = \bma{@{}c@{\,\,}c@{\,\,\,\,}c@{}}{u_{1,t}^{\top} & \hdots &  u_{M,t}^{\top}}^\top,
\end{equation}

For each subsystem $i \in \N$, we define the set of neighboring subsystems $\Ni \subseteq \N$ which contains all those subsystems that are coupled to  subsystem $i$ over the dynamics, constraints or cost. We define the state vector $x_{\Ni,t} \in \mathbb{R}^{n_{\Ni}}$ containing the local states of subsystem $\agentIdxA$ and its neighboring subsystems in $\Ni$, which can be expressed as $x_{\Ni,t} = X_{\Ni} x_{t}$, with $X_{\Ni}$ being a projection matrix, i.e., a binary matrix $X_{\Ni} \in \{0,1\}^{n_{\Ni}\times n}$. 
Similarly, the projection matrices $X_i \in \{0,1\}^{n_{i}\times n}$ and $U_i \in \{0,1\}^{m_{i}\times m}$, are defined such that $x_{i,t} = X_i x_t$ and $u_{i,t} = U_i u_t$. The dynamics of subsystem $i$ is then given as 
\begin{equation}\label{eq:sys_loc}
x_{i,t+1} = \sum_{j \in \Ni} A_{ij} x_{j,t} + B_i u_{i,t}, \quad \forall i \in \mathcal \N,  
\end{equation} 
with 
\begin{equation}
A_{\Ni} = X_i A X_{\Ni}^\top, \quad B_i = X_i B U_i^\top. 
\end{equation}
The local state and input constraints are defined as 
\begin{equation}\label{eq:const_loc}
\begin{aligned} 
&x_{\Ni,t} \in \mathcal{X}_{\Ni} = \{ x_{\Ni,t} \in \mathbb{R}^{n_{\Ni}}: G_{\Ni} x_{\Ni} \leq g_{\Ni} \}, \\
&u_{i,t} \in \mathcal{U}_{i} = \{u_{i,t} \in \mathbb{R}^{m_i}: L_{i} u_i \leq l_i \}, 
\end{aligned}
\end{equation}
with $L_{i} = U_i L U_i^\top$ and $l_{i} = U_i l$, and  
$G_{\Ni} = X_{\Ni} G X_{\Ni}^\top$ and $g_{\Ni} = X_{\Ni} g$.

\subsection{Control Problem Formulation} 
Let us consider system \eqref{eq:sys_glob}. We are given an iterative task, where the trajectories of the subsystems start at the same initial states at each iteration. We will discuss the case of non-iterative tasks in Section~\ref{sec:DE}.  
In the following, we denote the iteration by a superscript $\iter$ and the initial state at iteration $\iter$ by 
\begin{equation}
x_{i,0}^\iter = x_{i,S}, \quad \forall i \in \N, 
\end{equation}
where the overall initial state $x_0^\iter = x_S$ is defined as a stacked vector similar to \eqref{eq:xu_loc}.

The goal is to solve the following infinite horizon optimal control problem (IHOCP) at each iteration  
\begin{equation}\label{eq:IHOCP}
\begin{aligned}
J_{0 \rightarrow \infty}^*(x_{S}) = \min_{u_0, u_1, ...} 
\quad & \sum_{t=0}^{\infty} 
h(x_{t},u_{t})  \\
\text{s.t.} \quad \ & \ x_{t+1} = A x_{t} + B u_{t}, \ \forall t \geq 0 \\
&\ x_{t} \in \mathcal{X}, \ \forall t \geq 0 \\
&\ u_{t} \in \mathcal{U}, \ \forall t \geq 0 \\
&\ x_{0} = x_S.
\end{aligned}
\end{equation}

In the following, we consider problems that involve decomposable stage costs in \eqref{eq:IHOCP}, i.e., where $h(x_t,u_t)$ is given as a sum of local stage costs $h_i(x_{\Ni,t},u_{i,t})$ as 
\begin{equation}\label{eq:J_glob}
h(x_t,u_t) = \sum_{i=1}^{M} h_i(x_{\Ni,t},u_{i,t}). 
\end{equation}
We assume that the local stage costs $h_i(\cdot,\cdot)$ are continuous, jointly convex and satisfy 
\begin{equation}
\begin{cases*}
h_i(x_{\Ni,F},0)=0, \\
h_i(x_{\Ni,t}^\iter,u_{i,t}^\iter) \geq 0, \,\,\,\,\,\, 
\forall \ x_{\Ni,t}^\iter \in \mathbb{R}^{n_{\Ni}} \!\! \setminus \!\!  \{x_{\Ni,F}\}, \,\, \\
\qquad \qquad \qquad \quad \qquad \forall \ u_{i,t}^\iter \in \mathbb{R}^{m_i} \!\! \setminus \!\! \{0\},
\end{cases*}
\end{equation} 
where the final state $x_F$ 
is a feasible equilibrium for system \eqref{eq:sys_glob} under no input, i.e., $A x_F = x_F$.

\begin{myrem}\label{rem:separablecost}
    While the local stage costs $h_i(x_{\Ni,t},u_{i,t})$ can account for coupling between the subsystems, 
this formulation includes the special case of completely separable cost functions with local stage costs given as $h_i(x_{i,t},u_{i,t})$. 
\end{myrem}

\begin{myrem}\label{rem:QR}
A specific choice of the stage cost $h(x_t,u_t)$ can be the quadratic function 
$$h(x_t,u_t) = x_{t}^\top Q x_{t} + u_{t}^\top R u_{t},$$ 
with positive semi-definite and positive definite weighting matrices $Q \in \mathbb{R}^{n \times n}$ and $R \in \mathbb{R}^{m \times m}$, respectively. In this case, the local stage costs are given by  
\begin{equation*}
h_i(x_{\Ni,t},u_{i,t}) = x_{\Ni,t}^\top Q_{\Ni} x_{\Ni,t} + u_{i,t}^\top R_i u_{i,t},  
\end{equation*} 
with $Q_{\Ni}$ and $R_{i}$ such that the global weighting matrices $Q$ and $R$ are given by 
$Q = \sum_{i \in \N} X_{\Ni}^\top Q_{\Ni} X_{\Ni}$ and 
$R = \sum_{i \in \N} U_{i}^\top R_{i} U_{i}$. 
A completely separable quadratic stage cost is then defined as 
\begin{equation}\label{eq:hi_loc_uncoup}
h_i(x_{i,t},u_{i,t}) = x_{i,t}^\top Q_i x_{i,t} + u_{i,t}^\top R_i u_{i,t}, 
\end{equation} 
with $Q_{i}$ such that $Q = \sum_{i \in \N} X_{i}^\top Q_{i} X_{i}$ and $R_i$ as before. 
\end{myrem}

\section{Background on LMPC}
\label{sec:LMPC}

We review the LMPC problem formulation for the global system in \eqref{eq:sys_glob} from \cite{Rosolia2017a}. 
For this, we define the vectors that collect all inputs applied to system \eqref{eq:sys_glob} and its resulting states for all time steps $t$ of iteration $\iter$ as   
\begin{equation} \label{eq:xu}
\begin{aligned}
\textbf{u}^\iter &= [u_0^{\iter \top}, u_1^{\iter \top}, ..., u_t^{\iter \top}, ...]^\top, \\
\textbf{x}^\iter &= [x_0^{\iter \top}, x_1^{\iter \top}, ..., x_t^{\iter \top}, ...]^\top. 
\end{aligned}
\end{equation} 

\subsection{Convex Safe Set}
In order to guarantee stability of MPC laws, an $N$-step controllable set to a control invariant set can be used. Computing such a set is usually numerically challenging or even intractable for nonlinear systems or large scale distributed systems. 
To alleviate this problem, we will as \cite{Rosolia2017} exploit previously seen trajectories that successfully completed the iterative task. Since they represent a subset of the maximal stabilizable set, the sampled safe set $\SS^\iter$ is defined over the realized trajectories of the system from previous iterations 
\begin{equation}\label{eq:SS}
\SS^\iter = \left\{ 
	\bigcup_{l \in \M^\iter} \bigcup_{t=0}^{\infty} x_t^l, 
	\right\}, 
\end{equation}
where $\M^\iter$ collects all iteration indices from previous successful iterations, i.e., which were feasible and converged to $x_F$, defined as 
\begin{equation}
\M^\iter = \left\{ l \in [0,\iter] : \lim\limits_{t \rightarrow \infty} x_t^l = x_F \right\}. 
\end{equation}

Because of the convexity of the constraints $\mathcal{X}$ and $\mathcal{U}$, any convex combination of the elements in the safe set $\SS^\iter$ is again a control invariant set for system \eqref{eq:sys_glob}, i.e., for any element in the convex safe set 
\begin{equation}\label{eq:CS}
\begin{aligned}
&\CS^\iter = \mathrm{conv}(\SS^\iter) \\
&= \left\{ \sum_{l\in T^\iter}\sum_{t=0}^{\infty}  \alpha_t^l x_t^l : \,\, \alpha_t^l \geq 0, \,\,  \sum_{l\in T^\iter}\sum_{t=0}^{\infty}  \alpha_t^l = 1, \,\, x_t^l \in \SS^\iter \right\}, 
\end{aligned}
\end{equation} 
there exists a sequence of control inputs that steers the system \eqref{eq:sys_glob} to $x_F$ \cite{Borrelli2003}. 
If all previous successful trajectories are taken into account, then it holds that the sets are growing over the iterations, i.e., $\M^{\iter-1} \subseteq \M^\iter$ and therefore 
\begin{equation}\label{eq:CS_growth}
	\CS^{\iter-1} \subseteq \CS^\iter. 
\end{equation} 

\subsection{Terminal Cost} 
For the $\iter$th realized trajectory $\textbf{x}^\iter$ and associated input sequence $\textbf{u}^\iter$ in \eqref{eq:xu}, the cost-to-go from time $t$ onwards is given by 
\begin{equation}
J_{t \rightarrow\infty}^\iter (x_t^\iter) = \sum_{k=t}^\infty h(x_{k}^\iter,u_{k}^\iter).
\end{equation}
The performance of the $\iter$th trajectory is defined as the cost from time $t=0$, i.e., 
\begin{equation}
J_{0\rightarrow\infty}^\iter(x_0^\iter) = \sum_{t=0}^\infty h(x_{t}^\iter,u_{t}^\iter). 
\end{equation}

The barycentric function \cite{Jones2010} is used as the terminal cost in the LMPC for linear systems in \cite{Rosolia2017}. 
It is defined as 
\begin{equation}\label{eq:cc}
\begin{aligned}
\V^{\iter,*}(x) = \min_{\boldsymbol{\ccv}^\iter} 
\quad &\sum_{\iterit=0}^{\iter} \sum_{t=0}^{\infty} \ccv_t^\iterit J_{t \rightarrow \infty}^\iterit (x_t^\iterit) \\
\mathrm{s.\,t.}~~~ &\sum_{\iterit=0}^{\iter} \sum_{t=0}^{\infty} \ccv_t^\iterit = 1 \\ 
&\sum_{\iterit=0}^{\iter} \sum_{t=0}^{\infty} \ccv_t^\iterit x_t^\iterit = x, \\
&\ccv_t^\iterit \geq 0, \quad \forall t \in \mathbb{N}, 
\end{aligned}
\end{equation}
with $x_t^\iterit$ being the realized state at time $t$ of the $\iterit$th iteration, and where 
$\boldsymbol{\ccv}^\iter$ comprises all $\ccv^\iterit_t, \forall \iterit \in \{0,...,\iter\}, \forall t \in \mathbb{N}$. 
The function $\V^{\iter,*}$ thus assigns to every point in the convex safe set the corresponding convex combination of minimum costs-to-go along the previous trajectories in the safe set. 

\begin{myrem} 
In practical applications, 
the iterations will have a finite time duration. For simplicity, we adopt the infinite time formulation in this paper. 
\end{myrem}

An LMPC \cite{Rosolia2017} for a centralized linear system then solves at each time step $t$ the following finite horizon optimal control problem (FHOCP), 
\begin{align}
J_{t \rightarrow t+N}^{\LMPC,\iter}(x_{t}^\iter) =\notag\\
\min_{\mathbf{x}_{t,N},\mathbf{u}_{t,N-1}, \boldsymbol{\alpha}^{\iter-1}} 
& \left[ 
\sum_{k=t}^{t+N-1}
h(x_{k|t},u_{k|t}) + \V^{\iter-1,*}(x_{k+N|t}) 
\right] \notag\\
\text{s.t.} \quad \quad & \ x_{k+1|t} = A x_{k|t} + B u_{k|t},  \notag\\
&\ x_{k|t} \in \mathcal{X},  \notag\\
&\ u_{k|t} \in \mathcal{U}, \ k = t, ..., t\!+\!N\!-\!1 \label{eq:LMPC_glob}\\
&\ x_{t|t} = x_t^\iter, \notag\\
&\ x_{t+N|t} \in \CS^{\iter -1},\notag
\end{align}
with 
\begin{equation}\label{eq:xu_glob}
\begin{aligned}
\mathbf{x}_{t,N} &= [x_{t|t}^\top,...,x_{t+N|t}^\top]^\top, \\
\mathbf{u}_{t,N-1} &= [u_{t|t}^\top,...,u_{t+N-1|t}^\top]^\top.
\end{aligned}
\end{equation}
Let us denote the optimal solution to \eqref{eq:LMPC_glob} by 
\begin{equation}\label{eq:xu_opt}
\begin{aligned}
\mathbf{x}_{t,N}^* &= [x_{t|t}^{*\top},...,x_{t+N|t}^{*\top}]^\top, \\
\mathbf{u}_{t,N-1}^* &= [u_{t|t}^{*\top},...,u_{t+N-1|t}^{*\top}]^\top. 
\end{aligned}
\end{equation}
At time $t$, the first input is applied to the system, i.e., $u_t^\iter = u_{t|t}^{*\iter}$, and the problem \eqref{eq:LMPC_glob} is solved again for the next time step in a receding horizon fashion. 

Under the assumption that at iteration $\iter = 1$ the convex safe set is non-empty, i.e., $\CS^{\iter - 1} = \CS^0 \neq \emptyset$, recursive and iterative feasibility, asymptotic stability and non-decreasing performance over the iterations are proved in \cite{Rosolia2017a}.

\section{DLMPC}
\label{sec:DLMPCformulation}

In the following, we present the problem formulation of DLMPC, 
which extends the LMPC approach to distributed systems. 

Let us consider the coupled constrained linear distributed system from \eqref{eq:sys_glob}. 
We define the vectors that collect all inputs applied to subsystem $\agentIdxA$ in \eqref{eq:sys_loc} and its resulting states for all time steps $t$ of iteration $\iter$ as   
\begin{equation} \label{eq:xu_loc_traj}
\begin{aligned}
\textbf{u}_\agentIdxA^\iter &= [u_{\agentIdxA,0}^{\iter \top}, u_{\agentIdxA,1}^{\iter \top}, ..., u_{\agentIdxA,t}^{\iter \top}, ...]^\top, \\
\textbf{x}_\agentIdxA^\iter &= [x_{\agentIdxA,0}^{\iter \top}, x_{\agentIdxA,1}^{\iter \top}, ..., x_{\agentIdxA,t}^{\iter \top}, ...]^\top. 
\end{aligned}
\end{equation} 
We further define the local sampled safe sets for subsystems $i \in \N$ over the realized trajectories of the subsystem from all successful previous iterations up to $\iter$ as 
\begin{equation}\label{eq:SS_i}
\SS_{\agentIdxA}^\iter = \left\{ 
\bigcup_{l \in \M^\iter} \bigcup_{t=0}^{\infty} x_{\agentIdxA,t}^l, 
\right\}, 
\end{equation}
with $\M^\iter$ as defined before for \eqref{eq:SS}. 
Moreover, we note that we can decompose the safe set from \eqref{eq:CS} into the following local convex safe sets
\begin{equation}\label{eq:CS_local} 
\begin{aligned}
\CS_{\agentIdxA}^\iter 
= \Bigg\{ &\sum_{l\in T^\iter}\sum_{t=0}^{\infty} \ccv_{\agentIdxA,t}^\iterit x_{\agentIdxA,t}^\iterit : \\ 
&\ccv_{\agentIdxA,t}^\iterit \geq 0, \,\,  \sum_{l\in T^\iter}\sum_{t=0}^{\infty} \ccv_{\agentIdxA,t}^\iterit = 1, \,\, x_{\agentIdxA,t}^\iterit \in \SS_\agentIdxA^\iter \Bigg\}, 
\end{aligned}
\end{equation} 
where the coefficients $\ccv_{\agentIdxA,t}^\iterit, \forall \iterit \in \{0,...,\iter-1\}, \forall t \geq 0$ will be optimized over in problems 
\eqref{eq:LMPC_loc} and \eqref{eq:consensus}.

We note the following relation of the convex safe set $\CS^\iter$ in \eqref{eq:CS} and the local convex safe sets $\CS_\agentIdxA^\iter$ in \eqref{eq:CS_local}, which will be important for the decomposition of the problem in \eqref{eq:LMPC_glob}:  
\begin{equation}\label{eq:CSCSi}
    \begin{aligned}
    &x = [x_1^\top, ... ,x_M^\top]^\top \in \CS^\iter \iff \\
    &x_\agentIdxA \in \CS_\agentIdxA^\iter, 
    \quad \boldsymbol{\alpha}_\agentIdxA^{\iter} = \boldsymbol{\alpha}_\agentIdxB^{\iter}, \forall \agentIdxA \neq \agentIdxB, \agentIdxA, \agentIdxB \in \N, 
    \end{aligned}
\end{equation}
with 
$\boldsymbol{\ccv}_{\agentIdxA}^{\iter}$ comprising $\ccv_{\agentIdxA,t}^\iterit, \forall \iterit \in \{0,...,\iter-1\}, \forall t \geq 0$ in \eqref{eq:CS_local}.

Based on the assumption before that the global system is decomposable into $\numAgent$ coupled subsystems, the global LMPC problem in \eqref{eq:LMPC_glob} can equivalently be decomposed into the following subproblems 
\begin{align}\label{eq:LMPC_loc}
J_{\agentIdxA,t \rightarrow t+N}^{\LMPC,\iter}&(x_{\agentIdxA,t}^\iter) = \notag\\ 
\min_{\substack{ \mathbf{x}_{\mathcal{N}_\agentIdxA,t,N}, \notag\\ 
		\mathbf{u}_{\agentIdxA,t,N-1}, \notag\\ 
		\boldsymbol{\ccv}_{\agentIdxA}^{\iter-1}}}
& \Bigg[
\sum_{k=t}^{t+N-1}
h_\agentIdxA(x_{\mathcal{N}_\agentIdxA,k|t},u_{\agentIdxA,k|t}) \, + 
\V^{\iter-1,*}_{\agentIdxA}(x_{\agentIdxA,t+N|t}) 
\Bigg] \notag\\
\text{s.t.} \ & \ x_{\agentIdxA,k+1|t} = A_{\mathcal{N}_\agentIdxA} x_{\mathcal{N}_\agentIdxA, k|t} + B_\agentIdxA u_{\agentIdxA,k|t},  \notag\\
&\ x_{\mathcal{N}_\agentIdxA,k|t} \in \mathcal{X}_{\mathcal{N}_\agentIdxA},  \notag\\
&\ u_{\agentIdxA,k|t} \in \mathcal{U}_\agentIdxA, \ k = t,...,t\!+\!N\!-\!1 \\
&\ x_{\mathcal{N}_\agentIdxA, t|t} = x_{\mathcal{N}_\agentIdxA,t},  \notag\\
&\ x_{\agentIdxA,t+N|t} \in \CS_\agentIdxA^\iter, \notag
\end{align}
with 
\begin{equation}\label{eq:LMPC_loc_vars}
\begin{aligned}
\mathbf{x}_{\mathcal{N}_\agentIdxA,t,N} &= [x_{\mathcal{N}_\agentIdxA,t|t}^\top,...,x_{\mathcal{N}_\agentIdxA,t+N-1|t}^\top]^\top, \\
\mathbf{u}_{\agentIdxA,t,N-1} &= [u_{\agentIdxA,t|t}^\top,...,u_{\agentIdxA,t+N-1|t}^\top]^\top, 
\end{aligned}
\end{equation}
with 
$\boldsymbol{\ccv}_{\agentIdxA}^{\iter-1}$ comprising  $\ccv_{\agentIdxA,t}^\iterit, \forall \iterit \in \{0,...,\iter-1\}, \forall t \in \mathbb{N}$,  
and with 
$\V^{\iter-1,*}_{\agentIdxA}(x_{\agentIdxA,k+N|t})$ 
being defined as in \eqref{eq:cc}, but with 
$J_{t \rightarrow \infty}^\iterit (x_t^\iterit)$ 
replaced by 
$$
J_{\agentIdxA, t \rightarrow\infty}^\iterit (x_{\agentIdxA, t}^\iterit) = \sum_{k=t}^\infty h_\agentIdxA(x_{\Ni,k}^\iterit,u_{\agentIdxA, k}^\iterit).
$$

In order to guarantee that the decomposed problem in \eqref{eq:LMPC_loc} is an exact reformulation of the global problem in \eqref{eq:LMPC_glob}, 
i.e., to guarantee that they have the same solutions, the following consensus constraints need to be introduced 
\begin{equation}\label{eq:consensus}
\begin{aligned} 
\boldsymbol{\ccv}_{\agentIdxA}^{\iter-1} &= \boldsymbol{\ccv}_{\agentIdxB}^{\iter-1}, \quad\quad\quad\quad\quad \forall \agentIdxA, \agentIdxB \in \N, \agentIdxA \neq \agentIdxB, \\
\textbf{x}_{\Ni, t, N} &= X_{\Ni} \, \textbf{x}_{t,N}, ~\quad\quad\quad \forall \agentIdxA \in \N, ~ t \geq 0, 
\end{aligned}
\end{equation}
with $\textbf{x}_{t,N}$ the planned state trajectory of the global system as defined in \eqref{eq:xu_glob}. 
The consensus constraint in the first line of \eqref{eq:consensus} ensures the condition in \eqref{eq:CSCSi}, and the one in the second line ensures that overlapping parts of state variables from neighboring subsystems in $\textbf{x}_{\Ni, t, N}$, i.e., variables of different subsystems that have the same physical meaning, are the same. 
The local FHOCPs in \eqref{eq:LMPC_loc} are solved in a receding horizon fashion, i.e., the first local inputs $u_{\agentIdxA,t}^\iter = u_{\agentIdxA,t|t}^{*\iter}$ are applied to the subsystems at time $t$. The next section presents a distributed solution method to solve the subproblems \eqref{eq:LMPC_loc}.

\section{Distributed Synthesis for DLMPC}
\label{sec:ADMM}

In this section, we present a distributed solution method for the local decomposed subproblems in \eqref{eq:LMPC_loc} coupled over the consensus constraints in \eqref{eq:consensus}. 
Various distributed optimization algorithms can be employed \cite{Bertsekas1989}. 
We propose a distributed solution scheme based on the alternating direction method of multipliers (ADMM) because of its fast convergence in practice \cite{Farokhi2014,Boyd2010}. 
A consensus algorithm involving a central coordinator \cite{Boyd2010} could be implemented, which requires communication to every subsystem and therefore might not be tractable in practice. 
We propose a scheme, where only nearest-neighbor communication and no global coordination is required. A similar scheme has been presented before for distributed controller synthesis of large-scale systems in \cite{Stuerz2019b}.

\subsection{Distributed Synthesis for DLMPC}

Let us define 
the local variable vector of subsystem $\agentIdxA$ as 
\begin{equation}
\lv_{\agentIdxA} = [\textbf{x}_{\Ni,t,N}^\top, \,\, \boldsymbol{\ccv}_{\agentIdxA}^{\iter-1\top}, \,\, \textbf{u}_{\agentIdxA,t,N-1}^\top]^\top,   
\end{equation}
and the projection matrices $E_{\agentIdxA \agentIdxB}$, which project $\lv_\agentIdxA$ onto those variables over which a consensus needs to be achieved between subsystem $\agentIdxA$ and its neighboring subsystems $\agentIdxB \in \Ni$, 
i.e., $E_{\agentIdxA \agentIdxB} \lv_\agentIdxA = E_{\agentIdxB \agentIdxA} \lv_\agentIdxB$ is the consensus constraint from \eqref{eq:consensus} for subsystem $\agentIdxA$ with $\agentIdxB$. 
The decomposed problem in \eqref{eq:LMPC_loc} and \eqref{eq:consensus} can now be formulated as the following $\numAgent$ subproblems for all $i \in \N$ 
\begin{equation}\label{eq:LMPC_loci}
\begin{aligned}
\min_{\lv_\agentIdxA} ~&J_{\agentIdxA,t \rightarrow t+N}^{\iter}(\lv_\agentIdxA) + g_\agentIdxA(\lv_\agentIdxA) 
\\
\mathrm{s.t.~} &E_{\agentIdxA \agentIdxB} \lv_\agentIdxA = E_{\agentIdxB \agentIdxA} \lv_\agentIdxB, \quad \forall \agentIdxB \in \Ni.  
\end{aligned}
\end{equation}
with 
$J_{\agentIdxA,t \rightarrow t+N}^{\iter}(\lv_\agentIdxA)$ being the cost function in \eqref{eq:LMPC_loc}, and 
$g_\agentIdxA(\lv_\agentIdxA)$ being the indicator function for the constraints in \eqref{eq:LMPC_loc}, i.e., 
\begin{equation*} 
g_{\agentIdxA}(\lv_\agentIdxA) = 
\begin{cases}
0 ~\qquad \text{if~} \lv_\agentIdxA \text{~satisfies the constraints in \eqref{eq:LMPC_loc}}, \\
+\infty \quad \text{otherwise}.
\end{cases}
\end{equation*}

In order to derive the ADMM steps, we formulate the augmented Lagrangian, which allows a decomposition into the following sum of local terms 
\begin{equation}
\begin{aligned}
\mathcal{L}_{\rho} = 
\sum_{\agentIdxA \in \N} \mathcal{L}_{\rho, \agentIdxA}, 
\end{aligned}
\end{equation}
with 
\begin{equation}
\begin{aligned}
\mathcal{L}_{\rho,i} = & ~J_{\agentIdxA,t \rightarrow t+N}^{\iter}(\lv_\agentIdxA)
+ g_i(\lv_i) 
\\ & + 
\sum_{\agentIdxB \in \Ni} \left( \dm_{{ij}}^\trans \left( E_{\agentIdxA \agentIdxB} \lv_\agentIdxA - E_{\agentIdxB \agentIdxA} \lv_\agentIdxB \right) + \frac{\rho}{2}   \| E_{\agentIdxA \agentIdxB} \lv_\agentIdxA - E_{\agentIdxB \agentIdxA} \lv_\agentIdxB \|_2^2 \right). 
\end{aligned}
\end{equation}
The modified ADMM update steps are summarized in Algorithm~\ref{alg:ADMM}. 
The derivation can be found in \cite{Stuerz2019b}. 
The update steps require communication only between neighboring subsystems, i.e., subsystems that are coupled through their dynamics, constraints, or costs. 

\begin{algorithm}[t!]
    \SetAlgoLined
	\caption{Distributed computation of local input $u_{\agentIdxA,t}^\iter = u_{\agentIdxA,t|t}^{*\iter}$ for subsystem $\agentIdxA$ at time $t$ of iteration $\iter$}\label{alg:ADMM} 
	\KwIn{Iteration $\iter$, time $t$, $\rho > 0$, 
		set of neighboring subsystems $\Ni$,  current subsystems states $x_{i,t}^\iter$, initial values $\lv_i^{(0)}$, $\forall i \in \N$}  
	\For{$\agentIdxA \in \N$}{
	\textbf{Initialization}: Set $\kappa \gets 0$, $\dm_\agentIdxA^{(0)} \gets 0$\;  
	\While{\textit{not converged}}{
		Communicate $E_{\agentIdxA\agentIdxB} \, \lv_\agentIdxA^{(\kappa)}$ to neighboring nodes $\agentIdxB \in \Ni$\; 
		$\dm_\agentIdxA^{(\kappa + 1)} \gets \dm_\agentIdxA^{(\kappa)} + \rho \sum_{\agentIdxB \in \Ni}(T_{\agentIdxA\agentIdxB} \, \lv_i^{(\kappa)} - T_{\agentIdxB\agentIdxA} \, \lv_\agentIdxB^{(\kappa)})$\; 
		\vspace{-0.2cm} 
		\begin{equation*} 
		\begin{aligned} 
		\hspace{-0.3cm} 
		\lv_\agentIdxA^{(\kappa + 1)} \!\!\gets\!\!  & ~\underset{\lv_\agentIdxA}{\mathrm{argmin}} \Bigg\{ \!\! J_{\agentIdxA,t \rightarrow t+N}^{\iter}(\lv_\agentIdxA) 
		\!+\! g_\agentIdxA(\lv_\agentIdxA) 
		\!+\!  \lv_\agentIdxA^\top \dm_\agentIdxA^{(\kappa+1)}    \\ & + \!\rho\! \sum_{\agentIdxB \in \Ni}\| T_{\agentIdxA\agentIdxB} \, \lv_\agentIdxA - \frac{T_{\agentIdxA\agentIdxB} \, \lv_\agentIdxA^{(\kappa)}+ T_{\agentIdxB\agentIdxA} \, \lv_\agentIdxB^{(\kappa)}}{2} \|^2 \Bigg\}; %
		\end{aligned} 
		\end{equation*} 
		}
		$\kappa \gets \kappa+1$\; 
		Set $\lv_i^{*} = \lv_i^{\kappa}$\; 
		Return $u_{\agentIdxA,t|t}^{*\iter}$ from $\lv_i^{*}$\; 
		}
		\KwOut{local inputs $u_{\agentIdxA,t}^\iter$, ~ $\forall i \in \N$} 
\end{algorithm}

The DLMPC for iterative tasks, with distributed solution of the subproblems by  Algorithm~\ref{alg:ADMM}, is given in  Algorithm~\ref{alg:DLMPC}. 
\begin{algorithm}[t!]
    \SetAlgoLined
	\caption{DLMPC}\label{alg:DLMPC} 
		\KwIn{Initial states $x_{\agentIdxA,S}$, target states $x_{\agentIdxA,F}$,   
		sets of neighboring subsystems $\Ni$, initial successful feasible trajectories  $\mathbf{x}_{\agentIdxA,t}^{0},~ \mathbf{u}_{\agentIdxA,t}^{0}$ with $\CS_{\agentIdxA}^0$ and $J_{\agentIdxA, t \rightarrow \infty}^0(x_{\agentIdxA,t}), \forall x_{\agentIdxA,t} \in \mathbf{x}_{\agentIdxA,t}^{0}$,  $\forall i \in \N$, $\iter_\mathrm{max}$} 
		\For{iteration $\iter = 1$ to $\iter_\mathrm{max}$}{
		\For{$\agentIdxA \in \N$}{
		\While{$x_{\agentIdxA,t} \neq x_{\agentIdxA,F}$}{
		Solve local problem \eqref{eq:LMPC_loci} via Algorithm~\ref{alg:ADMM}\;  
		Apply local input $u_{\agentIdxA,t}^\iter = u_{\agentIdxA,t|t}^{*\iter}$\; 
		Obtain local state $x_{\agentIdxA,t}^\iter$\; 
		}
		Update $\CS_\agentIdxA^\iter$ by adding $\mathbf{x}_{\agentIdxA,t}^{*}$\;  
		Compute and save $J_{\agentIdxA, t \rightarrow \infty}^\iter(x_{\agentIdxA,t}^\iter), \forall x_{\agentIdxA,t}^\iter \in \mathbf{x}_{\agentIdxA,t}^{*}$\;
		}
		$\iter = \iter + 1$\; 
		}
	    \KwOut{Closed-loop trajectories  $\mathbf{x}_{\agentIdxA,t}^{*}, \, \mathbf{u}_{\agentIdxA,t}^{*}, \, \forall i \in \N$} 
\end{algorithm}

\subsection{Properties of the DLMPC} 

Next, we present our main result on the properties of Algorithm~\ref{alg:DLMPC}. We make the following assumptions. 
\begin{myas}\label{as:nonempty}
We have access to feasible trajectories 
$\textbf{x}_\agentIdxA^\iter$ at iteration $\iter = 0$ converging to $x_{\agentIdxA,F}$ for all subsystems $i \in \N$, and therefore the convex safe sets at iteration $\iter=1$,  $\CS_\agentIdxA^{\iter-1}=\CS_\agentIdxA^{0}$, are non-empty. 
\end{myas}
\begin{myas}\label{as:ADMM}
We assume that the local cost functions 
$J_{\agentIdxA,t \rightarrow t+N}^{\iter}(\cdot) + g_i(\cdot)$ in \eqref{eq:LMPC_loci} 
	are closed, proper and convex for all subsystems $\agentIdxA \in \N$,   and 
that the unaugmented Lagrangian $$\mathcal{L}_{i} = J_{\agentIdxA,t \rightarrow t+N}^{\iter}(\lv_\agentIdxA) 
	+ g_i(\lv_i) 
	+ 
	\sum_{\agentIdxB \in \Ni}  \dm_{{ij}}^\trans \left( E_{\agentIdxA \agentIdxB} \lv_\agentIdxA - E_{\agentIdxB \agentIdxA} \lv_\agentIdxB \right)$$ has a saddle point, and that the ADMM update steps in Algortihm~\ref{alg:ADMM} are feasible. 
\end{myas}

In addition to the classical MPC properties, namely, persistent feasibility in each iteration, and asymptotic stability of the equilibria $x_{i,F}$, the following properties hold for the DLMPC in Algorithm~\ref{alg:DLMPC}. 

\begin{mythe}\label{the:properties}
	Consider system \eqref{eq:sys_glob}, with distributed structure \eqref{eq:sys_loc} and  \eqref{eq:const_loc}. Let Assumptions~\ref{as:nonempty} and \ref{as:ADMM} hold.  
	Then, the DLMPC in Algorithm~\ref{alg:DLMPC} has the following properties:  
	\begin{enumerate}
		\item The DLMPC is feasible for all $t \geq 0$ and at every iteration $\iter \geq 1$. The equilibrium points $x_{\agentIdxA,F}$ are asymptotically stable for the closed-loop coupled subsystems under the DLMPC law. 
		\item The iteration cost $J_{0 \rightarrow \infty}^{\iter}(x_S)$ of the closed-loop system does not increase with the iteration index $\iter$, i.e., $J^{\iter+1}_{0\rightarrow\infty} \leq J^{\iter}_{0\rightarrow\infty}$.  
		\item 
		If the closed-loop system under the DLMPC converges to the steady-state inputs  ${\textbf{u}_\agentIdxA^\infty = \lim\limits_{\iter \rightarrow \infty} \textbf{u}_\agentIdxA^\iter}$ and the related steady-state trajectories  ${\textbf{x}_\agentIdxA^\infty = \lim\limits_{\iter \rightarrow \infty} \textbf{x}_\agentIdxA^\iter}$, for all subsystems $i \in \N$,  
		and the conditions from \cite[Theorem~3]{Rosolia2017} are satisfied,  then,  $\textbf{u}_\agentIdxA^\infty$ and  $\textbf{x}_\agentIdxA^\infty$ are global optimal solutions for the IHOCP~\eqref{eq:IHOCP}. 
	\end{enumerate}	
\end{mythe}
\begin{proof} 
The properties of Theorem~\ref{the:properties} have been proven in \cite{Rosolia2017} and \cite{Rosolia2017a} for a single system. 
It therefore suffices to show 
that the proposed decomposed problem solved in Algorithm~\ref{alg:DLMPC} is an exact reformulation of the global centralized problem and that the distributed solution method in Algorithm~\ref{alg:ADMM} converges to the global optimal solution.

	It can easily be seen that the local subproblems in \eqref{eq:LMPC_loc} together with the consensus constraints in \eqref{eq:consensus}, and their reformulation into the subproblems in \eqref{eq:LMPC_loci} are exact reformulations of the global problem in \eqref{eq:LMPC_glob}. This follows from the decomposability of the cost function in \eqref{eq:J_glob} and the structure of the system in \eqref{eq:sys_loc} and  \eqref{eq:const_loc}, together with the definitions of $\CS_i^\iter$ and $\V_i^{\iter,*}$ in \eqref{eq:CS_local} and \eqref{eq:LMPC_loc} with the consensus constraints in \eqref{eq:consensus}. 
	For linear system dynamics and convex constraints, the problems (both the global and the local ones) are convex and therefore admit a global optimal solution. 

	The proposed distributed solution method in Algorithm~\ref{alg:ADMM} is equivalent to the update steps of consensus ADMM in \cite{Boyd2010}. This equivalence has been shown in the derivation of the steps of Algorithm~A.1 in \cite{Banjac2019}. 
	Under Assumption~\ref{as:ADMM}, 
	the residuals 
$E_{\agentIdxA \agentIdxB} \lv_\agentIdxA - E_{\agentIdxB \agentIdxA} \lv_\agentIdxB, ~ \forall \agentIdxB \in \Ni, ~ \forall \agentIdxA \in \N$ 
 in Algorithm~\ref{alg:ADMM} asymptotically converge to zero and the cost 
 $\sum_{\agentIdxA \in \N} (J_{\agentIdxA,t \rightarrow t+N}^{\iter}(\cdot) + g_i(\cdot))$ from \eqref{eq:LMPC_loci} asymptotically converges to the global optimal solution. This is true in each time step  
 and therefore 
$J_{0 \rightarrow\infty}^\iter(x_S) = \sum_{\agentIdxA \in \N} J_{\agentIdxA, 0 \rightarrow\infty}^\iter(x_{\agentIdxA,S})$ 
converges to the global optimal solution. With the previous results, this is equivalent to the global optimal solution of the global centralized problem \eqref{eq:LMPC_glob}. 
 
 Therefore, the proofs in \cite{Rosolia2017} and \cite{Rosolia2017a} for a single system can be applied to the global centralized problem and thus the properties in Theorem~\ref{the:properties} hold. 
\end{proof}

Note that the properties in Theorem~\ref{the:properties} hold for the global system in \eqref{eq:sys_glob}, i.e., for the ensemble of all coupled subsystems.  In particular, property 2) 
guarantees a decrease in the iteration cost $J_{0 \rightarrow\infty}^\iter(x_S) = \sum_{\agentIdxA \in \N} J_{\agentIdxA, 0 \rightarrow\infty}^\iter(x_{\agentIdxA,S})$ of the sum of costs of all subsystems over iterations, rather than a decrease in the iteration costs $J_{\agentIdxA, 0 \rightarrow\infty}^\iter(x_{\agentIdxA,S})$ of the individual subsystems. 
Similarly, the optimal cost function  $J^{\mathrm{LMPC}}_{t \rightarrow t+N}(\cdot) = \sum_{\agentIdxA \in \N} J^{\mathrm{LMPC}}_{\agentIdxA, t \rightarrow t+N}(\cdot)$, is a Lyapunov  function for the equilibrium point $x_F$ of the closed loop system \eqref{eq:sys_glob} rather than the individual cost functions $J^{\mathrm{LMPC}}_{\agentIdxA, t \rightarrow t+N}(\cdot)$ for the individual subsystems. 
Furthermore, Algorithm~\ref{alg:ADMM} enables a distributed implementation of the global terminal constraint set $\CS^{\iter - 1}$ on which no distributed structure is imposed. 
The approach presented in this paper therefore captures the couplings between the subsystems and thus reduces conservatism w.r.t.\ other approaches of distributed MPC in the literature which 
impose structure on the terminal cost or terminal constraint sets.

The size of the decomposed local FHOCPs in \eqref{eq:LMPC_loc} are of the size of the individual subsystems and are independent of the number of subsystems.  
Since only nearest-neighbor communication is required in Algorithm~\ref{alg:ADMM}, also the solution method scales well with the number of subsystems. 
The number of data points for the construction of the convex safe sets in \eqref{eq:CS_local} grows in each iteration with adding the most recent closed loop trajectories to the safe sets in \eqref{eq:SS_i}. 
In order to reduce the required computational effort, the set of data points can be truncated, i.e., not all previously 
seen data points need to be included in the safe sets in \eqref{eq:SS_i}. 
For example only the most recent trajectories, or only the previous trajectory can be chosen to be included. 

\subsection{Safe and Efficient Data Generation and Domain Enlargement of the DLMPC Policy}
\label{sec:DE}

In order to use Algorithm~\ref{alg:DLMPC} for a (possibly iterative) task, data from at least one set of successful feasible trajectories of the subsystems are required to construct the local terminal sets and cost functions, which guarantee the properties in  Theorem~\ref{the:properties}. 
While successful feasible trajectories might be easy to obtain for distributed systems in some applications, such as by locally or manually controlling multiple loosely coupled subsystems in a non-optimal way, in other applications, such as for tightly coupled subsystems with safety-critical constraints, these data might be difficult to generate. 
We therefore propose in the following a distributed algorithm which allows the safe and efficient generation of the data required for the  computation of the terminal sets and costs in Algorithm~\ref{alg:DLMPC}. 
We present this data generation method for a control task from given initial states $x_{\agentIdxA,0}^{\mathrm{des}}$ to the target states $x_{\agentIdxA,F}$ of the subsystems. 
Let us define the following FHOCP, which is similar to the one in \eqref{eq:LMPC_loc} except for a different cost function, and with the initial states $x_{\agentIdxA,0}$ being optimization variables 
\begin{align}\label{eq:DE_loc}
&\min_{\substack{ \mathbf{x}_{\mathcal{N}_\agentIdxA,t,N}, \\ 
		\mathbf{u}_{\agentIdxA,t,N-1}, \\ 
		\boldsymbol{\ccv}_{\agentIdxA}^{\iter-1}}} 
	\| x_{i,0}^\iter - x_{i,0}^{\mathrm{des}} \|_2^2 \notag\\
\text{s.t.} \ & \ x_{\agentIdxA,k+1|t} = A_{\mathcal{N}_\agentIdxA} x_{\mathcal{N}_\agentIdxA, k|t} + B_\agentIdxA u_{\agentIdxA,k|t}, \ k = 0,...,N\!-\!1  \notag\\
&\ x_{\mathcal{N}_\agentIdxA,k|t} \in \mathcal{X}_{\mathcal{N}_\agentIdxA}, \ k = t,...,t\!+\!N\!-\!1 \notag\\
&\ u_{k|t}^{\agentIdxA} \in \mathcal{U}_\agentIdxA, \ k = t,...,t\!+\!N\!-\!1, \\
&\ x_{\agentIdxA,t+N|t} \in \CS_\agentIdxA^{\iter-1},\notag 
\end{align}
with $\mathbf{x}_{\mathcal{N}_\agentIdxA,t,N}$,
$\mathbf{u}_{\agentIdxA,t,N-1}$, $\boldsymbol{\ccv}_{\agentIdxA}^{\iter-1}$ as in \eqref{eq:LMPC_loc_vars}, and where the consensus constraints in \eqref{eq:consensus} have to hold.  
Note that no initial successful feasible trajectories $\mathbf{x}_{\agentIdxA,t}^{0}$ need to be available. Instead, we use only the target states $x_{\agentIdxA,F}$ as initial feasible trajectories and therefore define $\CS_\agentIdxA^0 = x_{\agentIdxA,F}$.

Iteratively solving \eqref{eq:DE_loc} and computing the DLMPC closed-loop trajectories by Algorithm~\ref{alg:DLMPC} enlarges the domain of the DLMPC policy and converges to feasible trajectories starting at $x_{\agentIdxA,0} = x_{\agentIdxA,0}^{\mathrm{des}}$ and ending in $x_{\agentIdxA,0} = x_{\agentIdxA,F}$, which can used as the input to Algorithm~\ref{alg:DLMPC}. 
These steps are summarized in the following Algorithm~\ref{alg:domain}

\begin{algorithm}[t!]
    \SetAlgoLined
	\caption{Efficient and safe distributed data generation for the DLMPC policy}\label{alg:domain} 
		\KwIn{ 
		Terminal states $x_{i,F}$,  
		sets of neighboring systems $\Ni$, desired initial states $x_{i,0}^{\mathrm{des}}$,  $\forall i \in \N$, $r_{\mathrm{max}}$} 
		\textbf{Initialize:} 
		Set $\CS_\agentIdxA^0 = x_{i,F}$\;
		Set iteration count $\iterDE = 1$\; 
		\For{$\agentIdxA \in \N$}{
		\While{$\| x_{i,0}^\iterDE - x_{i,0}^{\mathrm{des}} \|_2^2 \leq \epsilon$ \text{and} $r \leq r_{\mathrm{max}}$}{
		Solve \eqref{eq:DE_loc} to obtain $x_{i,0}^\iterDE$\;  
		Compute $\mathbf{x}_{\agentIdxA,t}^{*},~ \mathbf{u}_{\agentIdxA,t}^{*}$ via Algorithm~\ref{alg:DLMPC} 
		with inputs: 
		$x_{\agentIdxA,S} = x_{i,0}^\iterDE$, 
		$\CS_{\agentIdxA}^0 = \CS_\agentIdxA^{\iterDE-1}$, 
		$J_{\agentIdxA, t \rightarrow \infty}^0 = J_{\agentIdxA, t \rightarrow \infty}^\iterDE$, $\iter_\mathrm{max} = 1$, 
		until line~7 of Algorithm~\ref{alg:DLMPC}, then break\;   
		Update $\CS_\agentIdxA^\iterDE$ by $\mathbf{x}_{\agentIdxA,t}^{*}$\; 
		Compute and save $J_{\agentIdxA, t \rightarrow \infty}^\iterDE(x_{\agentIdxA,t}), \forall x_{\agentIdxA,t} \in \mathbf{x}_{\agentIdxA,t}^{*}$\;  
		$\iterDE = \iterDE + 1$\;  
		}
		Set $\CS_\agentIdxA = \CS_\agentIdxA^\iterDE$\; 
		}
		\KwOut{
		$\CS_\agentIdxA$ containing successful feasible trajectories from $x_{i,0}^{\mathrm{des}}$ to $x_{i,F}$, ~ $\forall i \in \N$} 
\end{algorithm}

\begin{myrem}
	Algorithm~\ref{alg:domain} can also be used to compute a larger domain of the DLMPC policy. 
	If no specific initial states $x_{\agentIdxA,0}^\mathrm{des}$ are given, 
	instead of the cost function $\| x_{i,0}^\iter - x_{i,0}^{\mathrm{des}} \|_2^2$ in \eqref{eq:DE_loc}, a different cost function can be used, for example to compute the initial states $x_{\agentIdxA, 0}$ for all subsystems $\agentIdxA \in \N$ as the points furthest in the direction of interest at the borders of the convex safe sets $\CS_\agentIdxA^\iter$. 
\end{myrem}

\section{Numerical Experiments}
\label{sec:numerics}

In this section, we present numerical examples to demonstrate the methods of the DLMPC scheme in Algorithm~\ref{alg:DLMPC} and the data generation in Algorithm~\ref{alg:domain}. 

We consider a system of three dynamically coupled subsystems with coupled state constraints. 
The subsystems have two states each, i.e., $x_i = [x_{i1}, ~x_{i2}]^\top, ~ \forall i \in \N$. The overall system state is given by $x = [x_1^\top, ~ x_2^\top, ~ x_3^\top]^\top \in \mathbb{R}^6$, and the input vector by $u = [u_1, ~ u_2, ~ u_3]^\top \in \mathbb{R}^3$. 
The system matrices of the global system are given by 
\begin{equation} \label{eq:sys_ex_1_2}
\begin{aligned}
A = \small\bma{@{}c@{\,\,\,\,\,\,}c@{\,\,\,\,\,\,}c@{}}{A_{11} & A_{12} & 0 \\ 
	0 & A_{22} & A_{23} \\ 
	A_{31} & 0 & A_{33} }, \quad
B = \diag{B_{11}, B_{22}, B_{33}}, 
\end{aligned}
\end{equation}
with 
\begin{equation}
\begin{aligned}
A_{11} &=  
\small\bma{@{}c@{\,\,\,\,\,\,}c@{}}{1 & 0.5 \\
	0 & 1.1}, 
~~ 
A_{22} = 
\small\bma{@{}c@{\,\,\,\,\,\,}c@{}}{1.05 & 0.6 \\
	0 & 1}, 
~~
A_{33} = 
\small\bma{@{}c@{\,\,\,\,\,\,}c@{}}{1 & 0.55 \\
	0 & 1.05}, 
\\ 
A_{12} &= A_{23} = A_{31} = -\small\bma{@{}c@{\,\,\,\,\,\,}c@{}}{0.1 & 0.2 \\
	0 & 0.3}, 
\end{aligned}
\end{equation} 
\begin{equation}
\begin{aligned}
-0.9 &\leq x_{11} - x_{21} \leq 0.9, \\ 
-0.9 &\leq x_{21} - x_{31} \leq 0.9, \\ 
-3 &\leq u_{i} \leq 3 \\ 
-5 &\leq x_{ik} \leq 5, ~~ \forall i \in \N, ~~k \in \{1,2\}. 
\end{aligned}
\end{equation}

\subsection{Data Generation}
\label{subsec:numerics_DE}

First, we generate the feasible trajectories required as inputs to Algorithm~\ref{alg:DLMPC}, by making use of Algorithm~\ref{alg:domain}. 
We choose the following desired initial states 
\begin{equation}\label{eq:initialstatesdes12}
\begin{aligned} 
x_{1,0}^\mathrm{des,1} &= [-5, 0]^\top, \qquad \qquad 
x_{1,0}^\mathrm{des,2} = [4, 0]^\top, \\
x_{2,0}^\mathrm{des,1} &= [-4.5, 0]^\top, \,\quad \qquad 
x_{2,0}^\mathrm{des,2} = [4.5, 0]^\top, \\
x_{3,0}^\mathrm{des,1} &= [-4, 0]^\top, \qquad \qquad 
x_{3,0}^\mathrm{des,2} = [5, 0]^\top. 
\end{aligned}
\end{equation}
We iteratively compute the inital states $x_{\agentIdxA,0}$ as those states closest to 
$x_{\agentIdxA,0}^\mathrm{des,1}$ and $x_{\agentIdxA,0}^\mathrm{des,2}$, in an alternating way, thus enlarging the domain of the DLMPC policy in both the negative and positive $x_{\agentIdxA1}$ directions within the feasible region of the state space. 
\begin{myrem}
	Note that if no specific initial states $x_{\agentIdxA,0}^\mathrm{des,1}$ and $x_{\agentIdxA,0}^\mathrm{des,2}$ are defined, a similar result of domain enlargement is achieved by changing the cost function in \eqref{eq:DE_loc} to $x_{i,0}^\iter$ 
	and  
	$-x_{i,0}^\iter$, 
	respectively. 
\end{myrem}

Figure~\ref{fig:DE1} shows the enlargement of the convex hulls of the safe sets, $\CS_i^\iter$, of the three subsystems over iterations $\iter=0$ to 4 of Algorithm~\ref{alg:domain}. At iteration 4, the given initial states $x_{i,0}^\mathrm{des,1}$ and $x_{i,0}^\mathrm{des,2}$ have been reached, i.e., closed-loop trajectories from  $x_{i,0}^\mathrm{des,1}$ to $x_{i,F}$ and from $x_{i,0}^\mathrm{des,2}$ to $x_{i,F}$ under the DLMPC law have been generated. 
\colorlet{mygreen}{black!20!green}
\colorlet{myblue}{black!20!blue}
\colorlet{myred}{black!20!red}
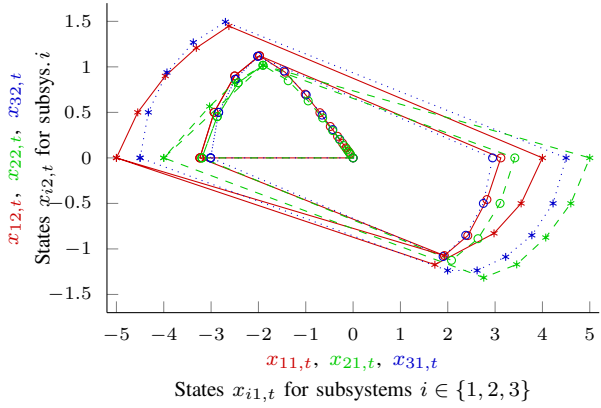
\begin{figure}[h] 
	\centering\input{figures/A_DE_x_23_CSi_new_xF_BETTER_tikz.tex} 
	\caption{ \label{fig:DE1} 
		Domain enlargement for desired initial states $x_{i,0}^{\mathrm{des,1}}$, $x_{i,0}^{\mathrm{des,2}}$:  
		$\CS_i^\iter$ of subystems $i=1$ 
		(solid red \,\textcolor{myred}{-----}\,), 
		$i=2$ (dashed green \,\textcolor{mygreen}{- - -}\,), 
		$i=3$ (dotted blue \,\textcolor{myblue}{......}\,) over 
		iterations $\iter=1,2$ (\,\textbf{$\circ$}\,), $\iter=3,4$ (\,\textbf{$\ast$}\,), starting with $\CS_i^0 = [0,\,0]^\top$. 
		}
\end{figure}

\subsection{Iterative Control Task}
We consider now the control task to steer the coupled subsystems in \eqref{eq:sys_ex_1_2} from $x_{\agentIdxA,0}^{\mathrm{des,1}}$ as in \eqref{eq:initialstatesdes12} to $x_{\agentIdxA,F} = [0,\,0]^\top$. 
The cost function is given as in \eqref{eq:hi_loc_uncoup} with $Q_i = I_{n_i}$ and $R_i = I_{m_i}$. 
We generate the first feasible trajectory by a distributed MPC with time horizon $N=15$, without terminal sets and constraints, and with $Q_i = 0.1 I_{n_i}$ and $R_i = I_{m_i}$. 
In real, possibly safety-critical, applications a feasible trajectory could be obtained by manual control of the subsystems, or with the data generation in Algorithm~\ref{alg:domain}, as illustrated in Section~\ref{subsec:numerics_DE}. 
Table~\ref{tab:IT2} shows the performance improvement over iterations of Algorithm~\ref{alg:DLMPC} with a time horizon of $N=4$. Figure~\ref{fig:IT2} shows the resulting closed-loop trajectories. 
It is interesting to note that while the iteration costs of subsystems 2 and 3 are decreasing over the iterations, the one of subsystem 1 is increasing. As noted before, Theorem~\ref{the:properties} guarantees that the sum, i.e., the iteration cost of the overall system is guaranteed to be non-increasing. 
\begin{table}
	\centering 
	\scriptsize
	\begin{tabular}{cccccc}
		\multicolumn{1}{c}{Iteration} &
		\multicolumn{1}{c}{0} & \multicolumn{1}{c}{1} & \multicolumn{1}{c}{2} & \multicolumn{1}{c}{3} & \multicolumn{1}{c}{4} 
		\\[-0.3cm] \\\hline \\[-0.2cm] 
		Sys. & 295.63 & 216.96 & 216.41 & 216.31 & 216.28 \\
		Subsys.\,1 & 112.47 & 87.97 & 88.07 & 88.22 & 88.31 \\
		Subsys.\,2 & 113.05 & 76.52 & 76.10 & 75.90 & 75.78 \\
		Subsys.\,3 & 70.11 & 52.47 & 52.23 & 52.19 & 52.19 \\[0.2cm]
		\multicolumn{1}{c}{5} &
		\multicolumn{1}{c}{6} & \multicolumn{1}{c}{7} & \multicolumn{1}{c}{8} & \multicolumn{1}{c}{9} & \multicolumn{1}{c}{10} 
		\\[-0.3cm] \\\hline \\[-0.2cm] 
		216.26 & 216.26 & 216.25 & 216.25 & 216.25 & $216.25^*$ \\
		88.36 & 88.38 & 88.40 & 88.41 & 88.42 & 88.42 \\
		75.72 & 75.68 & 75.66 & 75.64 & 75.63 & 75.63 \\
		52.19  & 52.19 & 52.19  & 52.20 & 52.20 & 52.20
	\end{tabular}
	\caption{Iteration Costs 
		$^*$converged to the global optimal solution (computed for the centralized system with $N=200$).}
	\label{tab:IT2}
\end{table}
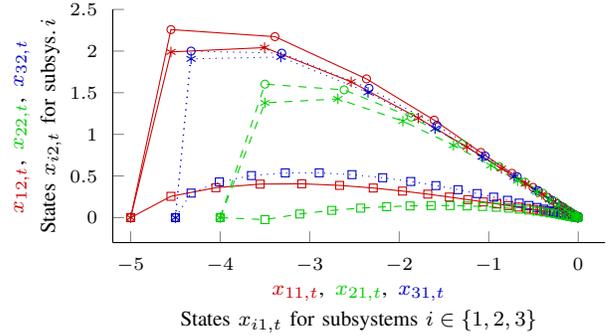
\begin{figure}
	\centering\input{figures/A_B_IT_x_21_tikz_paper_BETTER.tex} 
	\caption{ \label{fig:IT2} 
		Iterative regulation task for 3 dynamically coupled subsystems.  
		Trajectores of subystems $i=1$ 
		(solid red \,\textcolor{myred}{-----}\,), %
		$i=2$   (dashed green\,\textcolor{mygreen}{- - -}\,), 
		$i=3$   (dotted blue\,\textcolor{myblue}{......}\,) shown for  
		iterations $q=0$ (\,\textbf{$\square$}\,), $q=1$  (\,\textbf{$\circ$}\,), $q=10$  (\,\textbf{$\ast$}\,). }
\end{figure}

\section{Conclusion}
\label{sec:con}
A distributed learning model predictive control scheme was presented, which exploits data in order to construct the terminal cost and constraints of the DMPC problem without imposing the distributed structure of the system. The required computation is done online in a distributed way. 
It was shown how the scheme can be used to safely explore the state-space and generate the required data or exploit data from iterative control tasks. In addition to recursive feasibility and asymptotic stability,  performance improvement over iterations and convergence to the global centralized optimal solution under mild conditions are guaranteed.





\bibliographystyle{IEEEtran}
\bibliography{IEEEabrv,bibliography}

\end{document}

%% file: nomenclature.tex
\def\qed {{
		\parfillskip=0pt 
		\widowpenalty=10000 
		\displaywidowpenalty=10000 
		\finalhyphendemerits=0 
		\leavevmode 
		\unskip 
		\nobreak 
		\hfil 
		\penalty50 
		\hskip.2em 
		\null 
		\hfill 
		$\blacksquare$
		\par}}

\newtheorem{myrem}{Remark}
\newtheorem{myas}{Assumption}

\newtheorem{mythe}{Theorem}

\newcommand{\bma}[2]       {\left[\begin{array}{#1}#2\end{array}\right]}	
\newcommand{\mrm}[1]       {\mathrm{#1}}

\newcommand{\diag}[1]      {\mrm{{\diagonal}}\!\left(#1\right)}

\newcommand{\trans}        {\top}																				  

\newcommand{\V}{V}
\newcommand{\iterit}{l}

\newcommand{\iterDE}{r}

\newcommand{\N}[0]{\mathcal{N}} 
\newcommand{\Ni}[0]{\mathcal{N}_i} 
\newcommand{\iter}[0]{q} 
\renewcommand{\SS}[0]{\mathcal{SS}} 
\newcommand{\CS}[0]{\mathcal{CS}} 
\newcommand{\lv}{s} 
\newcommand{\dm}{\lambda} 
\newcommand{\ccv}{\alpha} 
\newcommand{\LMPC}{\mathrm{LMPC}}
\newcommand{\M}{T}

\newcommand{\agentIdxA}{i}
\newcommand{\agentIdxB}{j}

\newcommand{\numAgent}{M}

%% file: figures/A_DE_x_23_CSi_new_xF_BETTER_tikz.tex
%
%
\begin{tikzpicture}
\colorlet{mygreen}{black!20!green}
\colorlet{myblue}{black!20!blue}
\colorlet{myred}{black!20!red}

\begin{axis}[%
width=2.576772in,
height=1.621242in,
at={(0.767716in,0.492143in)},
scale only axis,
xmin=-5.2,
xmax=5.2,
xtick={-5, -4, -3, -2, -1,  0,  1,  2,  3,  4,  5},
xlabel={\parbox{6.5cm}{\centering
\textcolor{myred}{$x_{11,t}$},\, \textcolor{mygreen}{$x_{21,t}$},\, \textcolor{myblue}{$x_{31,t}$} \\[0.1cm]
States $x_{i1,t}$ for subsystems $i\in\{1,2,3\}$}
},
ymin=-1.7,
ymax=1.7,
ytick={  -3, -2.5,   -2, -1.5,   -1, -0.5,    0,  0.5,    1,  1.5,    2,  2.5,    3},
ylabel={\parbox{6.5cm}{\centering
\textcolor{myred}{$x_{12,t}$},\, \textcolor{mygreen}{$x_{22,t}$},\, \textcolor{myblue}{$x_{32,t}$} \\[0.1cm]
States $x_{i2,t}$ for subsys.\,$i$}
},
every tick label/.append style={font=\scriptsize},
yticklabel style = {xshift=-0.5ex},
xticklabel style = {yshift=-0.5ex},
xlabel style = {font=\footnotesize},
ylabel style = {font=\footnotesize},
ylabel style = {yshift=-1.5ex},
xlabel style = {yshift=0.5ex},
axis x line*=bottom,
axis y line*=left,
legend style={legend cell align=left,align=left,draw=white!15!black}
]
\addplot [color=black!20!red,solid,mark=o,mark options={solid},forget plot, mark size=1.5pt]
  table[row sep=crcr]{%
0	0\\
-3.23717987118578	0\\
-2.93621450804356	0.49999999985111\\
-2.5021441457275	0.899999999849395\\
-1.97634771808241	1.12172177540821\\
-1.43734679967697	0.941305968730995\\
-1.01182991120281	0.699663445081963\\
-0.701312395000164	0.496504718569322\\
-0.4815559482926	0.344964926036569\\
-0.328346126078673	0.2368961146124\\
-0.222577933401542	0.161408972740089\\
-0.150143971859417	0.10931406971776\\
-0.100882290078729	0.0736755368505015\\
-0.0675770239588253	0.0494662984538095\\
};
\addplot [color=black!20!blue,solid,mark=o,mark options={solid},dotted,forget plot, mark size=1.5pt]
  table[row sep=crcr]{%
0	0\\
-3.00965363142224	0\\
-2.84070362381659	0.49999999995575\\
-2.48624950520563	0.864142614000891\\
-2.01438502428611	1.11649235863636\\
-1.45831903040353	0.954839994658458\\
-0.990737476761788	0.691939770072587\\
-0.655365464148921	0.470162294959944\\
-0.429443362115485	0.311084885079533\\
-0.281071958710147	0.203935302500419\\
};
\addplot [color=black!20!green,solid,mark=o,mark options={solid},dashed,forget plot, mark size=1.5pt]
  table[row sep=crcr]{%
0	0\\
-3.19432689176762	0\\
-2.87060890464905	0.452857953182814\\
-2.42791557956436	0.825500850885859\\
-1.90367569697427	1.01738869891179\\
-1.37082149584619	0.847443206088857\\
-0.94925424627582	0.625902001100891\\
-0.643757843566442	0.438913930458642\\
-0.431524886028037	0.299799240697058\\
-0.287472694022709	0.201883908709405\\
-0.190981154547149	0.134921877233672\\
-0.126798123276493	0.0898460066614652\\
-0.0842312363702972	0.0597560594439265\\
-0.0560122820383651	0.0397463930666545\\
};
\addplot [color=black!20!red,solid,mark=o,mark options={solid},forget plot, mark size=1.5pt]
  table[row sep=crcr]{%
-3.23717987118578	0\\
-2.93621450804356	0.49999999985111\\
-2.5021441457275	0.899999999849395\\
-1.97634771808241	1.12172177540821\\
3.12157743856579	0\\
2.82691115428918	-0.457919698676393\\
2.42271971972479	-0.853711668546702\\
1.92791715679935	-1.07300621014665\\
-3.23717987118578	0\\
};
\addplot [color=black!20!blue,solid,mark=o,mark options={solid},dotted,forget plot, mark size=1.5pt]
  table[row sep=crcr]{%
-3.00965363142224	0\\
-2.84070362381659	0.49999999995575\\
-2.48624950520563	0.864142614000891\\
-2.01438502428611	1.11649235863636\\
2.94666284276611	0\\
2.75231585223928	-0.499999999988681\\
2.3794672865583	-0.850000000018702\\
1.90211182821734	-1.08371277142778\\
-3.00965363142224	0\\
};
\addplot [color=black!20!green,solid,mark=o,mark options={solid},dashed,forget plot, mark size=1.5pt]
  table[row sep=crcr]{%
-3.19432689176762	0\\
-2.87060890464905	0.452857953182814\\
-2.42791557956436	0.825500850885859\\
-1.90367569697427	1.01738869891179\\
3.41680132665141	0\\
3.10464358279483	-0.499999999898748\\
2.63853640715688	-0.8876240902902\\
2.07881351923413	-1.12502775723375\\
-3.19432689176762	0\\
};
\addplot [color=black!20!red,solid,mark=asterisk,mark options={solid},forget plot, mark size=1.5pt]
  table[row sep=crcr]{%
1.92791715679935	-1.07300621014665\\
-4.99975572837407	0\\
-4.54975572846743	0.499999999998057\\
-3.96725572857324	0.899999999996295\\
-3.31172550108059	1.20845615586098\\
-2.62371152518317	1.44779023375646\\
4.00000000015621	0\\
3.55000000020711	-0.499999999984318\\
2.97749774686956	-0.828738814327315\\
1.72791715679935	-1.17300621014665\\
-4.99975572837407	0\\
};\label{15}
\addplot [color=black!20!blue,solid,mark=asterisk,mark options={solid},dotted,forget plot, mark size=1.5pt]
  table[row sep=crcr]{%
-4.49999999906636	0\\
-4.32499999894963	0.49999999999901\\
-3.93226123581877	0.938479480436875\\
-3.375624849458	1.26888293489436\\
-2.69733304658831	1.49403691514993\\
4.499999999491	0\\
4.22502253345228	-0.499999999999229\\
3.77629619410961	-0.849999999911379\\
3.22313353778497	-1.08749999993378\\
2.62016362083188	-1.23646149346097\\
1.99635172994774	-1.23689052332632\\
-4.49999999906636	0\\
};\label{25}
\addplot [color=black!20!green,solid,mark=asterisk,mark options={solid},dashed,forget plot, mark size=1.5pt]
  table[row sep=crcr]{%
-1.90367569697427	1.01738869891179\\
-4.00000000070051	0\\
-3.03226123582069	0.565321818462405\\
-2.5046086628083	0.823587909374282\\
-1.90367569697427	1.01738869891179\\
4.99977466013269	0\\
4.59977466011707	-0.499999999979803\\
4.06977466010433	-0.874999999635682\\
3.45652264848321	-1.17012835483126\\
2.76059305802525	-1.31558064976689\\
-4.00000000070051	0\\
};\label{35}
\addplot [color=black!20!red,solid,mark=s,mark options={solid},forget plot, mark size=1.5pt]
table[row sep=crcr]{%
	0	0\\
	0	0\\
};
\addplot [color=black!20!blue,solid,mark=s,mark options={solid},dotted,forget plot, mark size=1.5pt]
table[row sep=crcr]{%
	0	0\\
	0	0\\
};
\addplot [color=black!20!green,solid,mark=s,mark options={solid},dashed,forget plot, mark size=1.5pt]
table[row sep=crcr]{%
	0	0\\
	0	0\\
};
\end{axis}
\end{tikzpicture}%

%% file: figures/A_B_IT_x_21_tikz_paper_BETTER.tex
%
%
\begin{tikzpicture}
\colorlet{mygreen}{black!20!green}
\colorlet{myblue}{black!20!blue}
\colorlet{myred}{black!20!red}

\begin{axis}[%
width=2.576772in,
height=1.221242in,
at={(0.767716in,0.492143in)},
scale only axis,
xmin=-5.2,
xmax=0.3,
xtick={-5, -4, -3, -2, -1,  0,  1},
xlabel={\parbox{6.5cm}{\centering
\textcolor{myred}{$x_{11,t}$},\, \textcolor{mygreen}{$x_{21,t}$},\, \textcolor{myblue}{$x_{31,t}$} \\[0.1cm]
States $x_{i1,t}$ for subsystems $i\in\{1,2,3\}$}
},
ymin=-0.3,
ymax=2.5,
ytick={  0, 0.5,   1, 1.5,   2, 2.5},
ylabel={\parbox{6.5cm}{\centering
\textcolor{myred}{$x_{12,t}$},\, \textcolor{mygreen}{$x_{22,t}$},\, \textcolor{myblue}{$x_{32,t}$} \\[0.1cm]
States $x_{i2,t}$ for subsys.\,$i$}
},
every tick label/.append style={font=\scriptsize},
yticklabel style = {xshift=-0.5ex},
xticklabel style = {yshift=-0.5ex},
xlabel style = {font=\footnotesize},
ylabel style = {font=\footnotesize},
ylabel style = {yshift=-1.5ex},
xlabel style = {yshift=0.5ex},
axis x line*=bottom,
axis y line*=left,
legend style={legend cell align=left,align=left,draw=white!15!black}
]
\addplot [color=black!20!red,solid,mark=square,mark options={solid},forget plot, every mark/.append style={mark size=1.5pt}]
  table[row sep=crcr]{%
-5	0\\
-4.55	0.255728262512169\\
-4.04854524400616	0.359374899169466\\
-3.55358371561262	0.402344553655119\\
-3.08830372849199	0.40633122115808\\
-2.66539475761325	0.387052022569343\\
-2.29012616548875	0.355478237941922\\
-1.96280678296264	0.318919018504159\\
-1.68068088490772	0.281948887705254\\
-1.43933060366719	0.2471771409258\\
-1.23366292586437	0.215864891446262\\
-1.05856110213973	0.18840208277383\\
-0.909276567085701	0.164661425416144\\
-0.78162980483464	0.144248753024797\\
-0.672078558515871	0.126669901294353\\
-0.577700770025659	0.1114332891083\\
-0.496128765471343	0.0981053901913036\\
-0.425461261483192	0.0863336715999417\\
-0.364171238917181	0.0758487342370753\\
-0.311020825772069	0.0664546098963363\\
-0.264989061074134	0.0580136541012246\\
-0.225214643771481	0.0504303401923591\\
-0.190953289036092	0.0436365545710755\\
-0.161547865098845	0.0375797103115616\\
-0.136408810997013	0.0322140961015407\\
-0.0968430688377322	0.0233772206002838\\
-0.0685527926359606	0.0167456189145083\\
-0.0485382485363228	0.01190731830446\\
-0.0344540325481112	0.00845348501579934\\
-0.024542435662583	0.00601945666068104\\
-0.0175376585573639	0.00430953379820398\\
};\label{11}
\addplot [color=black!20!blue,solid,mark=square,mark options={solid},dotted,forget plot, every mark/.append style={mark size=1.5pt}]
  table[row sep=crcr]{%
-4.5	0\\
-4.325	0.294546876311228\\
-4.00962154842201	0.428440380166982\\
-3.65082723404137	0.504875065555337\\
-3.27254641893916	0.537556407971078\\
-2.89344917430418	0.538011682952949\\
-2.5275226245415	0.51585999449497\\
-2.18462200318428	0.478979057577935\\
-1.87104616455116	0.433643895336099\\
-1.5901393713484	0.384674148974577\\
-1.34289760178555	0.335601910885231\\
-1.12855018951969	0.288857626424288\\
-0.945089867900127	0.245964686235118\\
-0.789731487043757	0.207731394490028\\
-0.659288506650106	0.174430064109871\\
-0.550464761443879	0.145955580721113\\
-0.460065812727813	0.121958861901406\\
-0.385138983749859	0.101953558044733\\
-0.323053931223445	0.0853967354788487\\
-0.27153660050917	0.0717460015057494\\
-0.22866904946932	0.0604965734744551\\
-0.192866336290366	0.0512022449491401\\
-0.162839846265823	0.0434841898743619\\
-0.137554409081883	0.037031209810686\\
-0.116184576133158	0.0315944995879605\\
-0.0826975720771827	0.0230359210697344\\
-0.0585411046392755	0.0167324163636966\\
-0.0411689941936967	0.0120489298291426\\
-0.0287870930706739	0.00858724865501515\\
-0.0200622915492404	0.00606553371600782\\
-0.0139791861310974	0.00426144630206515\\
};\label{21}
\addplot [color=black!20!green,solid,mark=square,mark options={solid},dashed,forget plot, every mark/.append style={mark size=1.5pt}]
  table[row sep=crcr]{%
-4	0\\
-3.5	-0.0245016289562683\\
-3.10962154842838	0.0437549557064347\\
-2.75257677822312	0.0868027017561541\\
-2.429945831427	0.114019311717049\\
-2.13967108136504	0.130875547151955\\
-1.878560459184	0.140476450155396\\
-1.64338144263804	0.144593200091166\\
-1.43135830799247	0.14431533285296\\
-1.24030656397362	0.140432038259187\\
-1.06857131074951	0.133629926254776\\
-0.914881537012199	0.124574492301202\\
-0.778189872587331	0.113924623742142\\
-0.657535957903811	0.102314388692792\\
-0.551949814244271	0.0903242467580876\\
-0.460397602934606	0.0784546096893424\\
-0.381764148424562	0.0671080618447766\\
-0.314862915901061	0.0565821235151942\\
-0.258463356139374	0.0470717481290004\\
-0.21132651762412	0.0386793580027627\\
-0.172241710124661	0.0314297472236005\\
-0.140059173864513	0.0252872875843693\\
-0.113715769354433	0.0201733068141436\\
-0.0922524326172601	0.0159820787265925\\
-0.0748234448700621	0.0125944443933683\\
-0.0492585510942346	0.00774794507029419\\
-0.0324648596852934	0.0047508937965608\\
-0.0213564229885381	0.002918849257175\\
-0.013939639131647	0.00178638644545703\\
-0.00895840296591636	0.00106778481073189\\
-0.00561947705914954	0.000600427389649383\\
};\label{31}
\addplot [color=black!20!red,solid,mark=o,mark options={solid},forget plot, every mark/.append style={mark size=1.5pt}]
  table[row sep=crcr]{%
-5	0\\
-4.55	2.25815833471359\\
-3.38835911113622	2.17339233331263\\
-2.36542834665294	1.66581521958815\\
-1.60937421063485	1.1707118696815\\
-1.08776918207896	0.78254918142944\\
-0.740434302801655	0.531711883616661\\
-0.504628922121195	0.36323540943708\\
-0.343039616998958	0.243209038977531\\
-0.233841403222856	0.166754519755286\\
-0.158590364538062	0.115139158267473\\
-0.106588850055031	0.076934132087013\\
-0.0717330239334785	0.051754768838089\\
-0.0482843012554892	0.0352451727444893\\
-0.0323696362603125	0.0235136991914799\\
-0.0217561679052817	0.015851409123849\\
-0.0146062215738934	0.0107923525058288\\
-0.00974383531113772	0.00724319983232147\\
-0.00647711225657233	0.00488705286369817\\
-0.00427244399155362	0.00328989007309576\\
};\label{12}
\addplot [color=black!20!blue,solid,mark=o,mark options={solid},dotted,forget plot, every mark/.append style={mark size=1.5pt}]
  table[row sep=crcr]{%
-4.5	0\\
-4.325	1.99969139246506\\
-3.31203318937195	1.97484360555116\\
-2.33806355782713	1.55329914779349\\
-1.5779224628955	1.10771576287208\\
-1.03925246393096	0.739324789152535\\
-0.681684444999802	0.491095028139251\\
-0.446837721529649	0.32356085874634\\
-0.293975204277703	0.209019130702171\\
-0.19584063605749	0.138551423992989\\
-0.131144184710991	0.0934124156090229\\
-0.0876454798945679	0.0618789395570528\\
-0.0587269875754667	0.0415068024930092\\
-0.0392663026661377	0.0281727582184085\\
-0.0260505324127355	0.0187421724099137\\
-0.0172492553153518	0.012503418810357\\
-0.011388691375282	0.00836329563843446\\
-0.00748408971535575	0.00551642916565457\\
-0.00491458361209169	0.00365158264019776\\
-0.00321716014624629	0.00241748809798194\\
};\label{22}
\addplot [color=black!20!green,solid,mark=o,mark options={solid},dashed,forget plot, every mark/.append style={mark size=1.5pt}]
  table[row sep=crcr]{%
-4	0\\
-3.5	1.60299012425496\\
-2.61498709860249	1.53416791088758\\
-1.86703730316322	1.20819473084718\\
-1.29915041044961	0.884891883294454\\
-0.885664827510478	0.61315357057431\\
-0.596163281772599	0.426711996703332\\
-0.393770630028932	0.291545874611348\\
-0.255604588667987	0.190693044270016\\
-0.165061260415089	0.125742486439646\\
-0.105869656502056	0.0828925048182307\\
-0.0674475740517173	0.0528567741278691\\
-0.0431042896932888	0.0340873808851603\\
-0.0275338815807205	0.0223897885120404\\
-0.0174401023224473	0.0144275498008622\\
-0.0109707261442378	0.00938048597400318\\
-0.00680612389277765	0.00612276771824039\\
-0.00413644999152185	0.00389895954756541\\
-0.00246627867571139	0.00247212402619906\\
-0.00143630980838431	0.00154942634831659\\
};\label{32}
\if01
\addplot [color=black!40!red,solid,mark=o,mark options={solid},forget plot, every mark/.append style={mark size=1.5pt}]
  table[row sep=crcr]{%
-5	0\\
-4.55	2.08368341339382\\
-3.46221191487001	2.10596561300425\\
-2.46727801485368	1.64719644047013\\
-1.7128604951198	1.18175442426585\\
-1.17810768521525	0.818626865829419\\
-0.808914164915058	0.559911051521703\\
-0.556055493681833	0.382306415364586\\
-0.382753367874027	0.261342459526475\\
-0.26369686007916	0.178810256620505\\
-0.181816179820927	0.122518069350085\\
-0.125463694190884	0.0840100478825683\\
-0.0866889856166895	0.0575673531619259\\
-0.0600426315495089	0.0394692231304574\\
-0.0417307861652442	0.0272117495809696\\
-0.0290894564949526	0.01933354031449\\
-0.020146498440983	0.0137383919794571\\
-0.0138245749462503	0.00963589188502007\\
-0.00940766474043285	0.00587479783322994\\
-0.00665652881828774	0.00400981283352614\\
-0.0047622896949112	0.00292456364792619\\
};
\addplot [color=black!40!blue,solid,mark=o,mark options={solid},forget plot, every mark/.append style={mark size=1.5pt}]
  table[row sep=crcr]{%
-4.5	0\\
-4.325	1.93276810783461\\
-3.32363562658444	1.95206234572116\\
-2.34906997664514	1.52043849082851\\
-1.59861465389088	1.07992933808731\\
-1.06952407678501	0.735361601465096\\
-0.711825438744125	0.491396992010195\\
-0.47456965891156	0.326540238828216\\
-0.318064038405414	0.217105629044559\\
-0.214647304040572	0.144945892280385\\
-0.145877522869469	0.0974715066597294\\
-0.0997098728351486	0.0660065132530245\\
-0.068435449740449	0.0449043374391362\\
-0.0471197084645739	0.0306736851371068\\
-0.0325224627711071	0.0210839569865195\\
-0.0224719269660852	0.0148550239994198\\
-0.0154189271757907	0.0104458260628742\\
-0.0104847133782303	0.0072475353725782\\
-0.00707565365873717	0.00446914180171785\\
-0.00493708007368946	0.00302187650377739\\
-0.00349097065863562	0.00217121644387236\\
};
\addplot [color=black!40!green,solid,mark=o,mark options={solid},forget plot, every mark/.append style={mark size=1.5pt}]
  table[row sep=crcr]{%
-4	0\\
-3.5	1.46023245642603\\
-2.65860883164445	1.48175429664311\\
-1.91861589960459	1.18107931435524\\
-1.35173376331786	0.870548346918782\\
-0.937997007853721	0.619209098921475\\
-0.643346608091268	0.431629021227359\\
-0.437041440229055	0.296970919340559\\
-0.294563168296481	0.201998789856351\\
-0.197456988993381	0.136055439472191\\
-0.13201886259986	0.0911163203901472\\
-0.0882268822732037	0.0608333972133905\\
-0.0590241539632642	0.0405375204845526\\
-0.0395730897674765	0.0270414447115814\\
-0.0265898766472474	0.0181635145652962\\
-0.0178692149360039	0.0126166987732674\\
-0.0118877930241096	0.00875557391892787\\
-0.00780525592049242	0.00597875713595764\\
-0.0050616603780947	0.00347647425427805\\
-0.00338379263084448	0.00229270873306283\\
-0.00225911251253637	0.00163567246348679\\
};
\fi
\if01 
\addplot [color=black!30!red,solid,mark=o,mark options={solid},forget plot, every mark/.append style={mark size=1.5pt}]
  table[row sep=crcr]{%
-5	0\\
-4.55	2.028552711094\\
-3.48678113873656	2.07348674596321\\
-2.50602190942694	1.64320104558751\\
-1.75239368776037	1.19384533094112\\
-1.21081209914019	0.836375727462609\\
-0.832403601451331	0.576701988689801\\
-0.571135658656975	0.394944710937105\\
-0.391629500906335	0.269755853796964\\
-0.268537788832201	0.184101969070171\\
-0.184201071114375	0.125611004169825\\
-0.126455779411258	0.0856877588062543\\
-0.0869444142290043	0.0584451306500634\\
-0.0599251146976698	0.0400833547355214\\
-0.0413686489656317	0.0277013919606885\\
-0.0285454804353689	0.0193320561511919\\
-0.0196127833183732	0.0132132612865243\\
-0.0135015840405764	0.00893039409527016\\
-0.00935799444305518	0.00555623045607263\\
-0.00672037442827953	0.00384934740812515\\
-0.0048784547929405	0.0028542573238288\\
};
\addplot [color=black!30!blue,solid,mark=o,mark options={solid},forget plot, every mark/.append style={mark size=1.5pt}]
  table[row sep=crcr]{%
-4.5	0\\
-4.325	1.9177874714178\\
-3.32326814750902	1.94155479211442\\
-2.34802436558074	1.51387368842634\\
-1.59788726410098	1.07564901630238\\
-1.06970407128784	0.733748865856152\\
-0.712331602935764	0.491581059220607\\
-0.474762176291365	0.327212076735245\\
-0.317736310397583	0.217799229320531\\
-0.213814826369536	0.145478747271064\\
-0.144681334147343	0.097641718982649\\
-0.09838326255923	0.0658542023839832\\
-0.0671733266304965	0.044602992283734\\
-0.0460157065707457	0.0304339114639861\\
-0.031569781037642	0.0209225277692284\\
-0.0216439690914664	0.0144886393387341\\
-0.01478973912077	0.0098720263877121\\
-0.0101141894565217	0.00666513197882995\\
-0.00694814128007052	0.00417654670633861\\
-0.00492648950225229	0.00287701509474385\\
-0.00353657876770115	0.00210451012035433\\
};
\addplot [color=black!30!green,solid,mark=o,mark options={solid},forget plot, every mark/.append style={mark size=1.5pt}]
  table[row sep=crcr]{%
-4	0\\
-3.5	1.4134531517985\\
-2.67331130872962	1.4542840841894\\
-1.93947429774444	1.17366661535725\\
-1.37199567747276	0.872557107552583\\
-0.955618965731025	0.624767720851617\\
-0.658190654841139	0.437908471126452\\
-0.449441033314417	0.302536873321194\\
-0.304921129309484	0.206816754876705\\
-0.205960134796056	0.140295141507872\\
-0.13876442189754	0.0946466764193149\\
-0.0934108435894445	0.063622533663197\\
-0.0629104238948112	0.0426827568422294\\
-0.0424294923386973	0.028782926253102\\
-0.0286230423768285	0.019557599505811\\
-0.019269776144207	0.0134186639619574\\
-0.0129013741518319	0.00899158313761801\\
-0.00863737735160952	0.00595529636590835\\
-0.00579788476535632	0.00358328829258522\\
-0.00400252285134345	0.00245113066158448\\
-0.00275223302626906	0.00179343048653548\\
};
\fi 
\if01 
\addplot [color=black!20!red,solid,mark=o,mark options={solid},forget plot, every mark/.append style={mark size=1.5pt}]
  table[row sep=crcr]{%
-5	0\\
-4.55	2.00358601552394\\
-3.497916770705	2.05403029187465\\
-2.52544874541389	1.63603485090409\\
-1.77420700134614	1.19642667048344\\
-1.23049217766412	0.843837739673213\\
-0.847847362616633	0.585180824537285\\
-0.582103324264247	0.402397405653158\\
-0.398823654444461	0.275458341828483\\
-0.272939193892519	0.188069497282508\\
-0.186714930788251	0.128173897820811\\
-0.127781820121727	0.0873084841148804\\
-0.0875416628971014	0.0595403658790387\\
-0.0600503072832927	0.0407739021175783\\
-0.0412065821860803	0.027962940666087\\
-0.0282768985680326	0.0191915452207498\\
-0.0194009107820358	0.0129442798665349\\
-0.0133944013000321	0.00867827402840376\\
-0.00934581332027874	0.00546249030473833\\
-0.00673976292282653	0.0038044775072584\\
-0.00491220336431119	0.00283480328178376\\
};
\addplot [color=black!20!blue,solid,mark=o,mark options={solid},forget plot, every mark/.append style={mark size=1.5pt}]
  table[row sep=crcr]{%
-4.5	0\\
-4.325	1.91104889233486\\
-3.32280363063936	1.93413741855074\\
-2.34781540613834	1.50778610999065\\
-1.59851385958352	1.07174948759026\\
-1.0710765754962	0.731908561693719\\
-0.713997011647266	0.491230375404917\\
-0.476253506568905	0.327721918318421\\
-0.31876161270444	0.218604358163714\\
-0.214308761356575	0.146206808363217\\
-0.144739000959744	0.0981386916992822\\
-0.098157913165681	0.066149380746912\\
-0.0667956873710094	0.0447919803140577\\
-0.0455658286826693	0.0304990441492183\\
-0.0311188722485302	0.0208183696992438\\
-0.0212619861231375	0.0142299171834596\\
-0.0145276523553629	0.00959197843399975\\
-0.00997169280089207	0.00643859157268875\\
-0.0068980854755884	0.004075016512379\\
-0.00491626612313171	0.00283152903713515\\
-0.00354466840950278	0.00208518872085264\\
};
\addplot [color=black!20!green,solid,mark=o,mark options={solid},forget plot, every mark/.append style={mark size=1.5pt}]
  table[row sep=crcr]{%
-4	0\\
-3.5	1.39091483020139\\
-2.68071404649403	1.43712724873427\\
-1.95130844099461	1.16555096616054\\
-1.38491750524574	0.870892330061164\\
-0.967791357674171	0.626454400799528\\
-0.668959767402662	0.440954231612459\\
-0.458686368661603	0.305986093321513\\
-0.312663166038977	0.210189997595194\\
-0.212267970282871	0.143328417907791\\
-0.143757320500836	0.0972045461380109\\
-0.0972581066102669	0.0657265882810913\\
-0.06579199786647	0.0444237245909554\\
-0.0445128562275422	0.0301273212200631\\
-0.0300925792516938	0.0204372500344825\\
-0.0203240216473378	0.0138645970043907\\
-0.0137091124822697	0.00921878068193899\\
-0.0092875480023066	0.00609858889247818\\
-0.00632953878712114	0.00375620079951703\\
-0.00442154507630657	0.00261830455063116\\
-0.00306839678262845	0.00193403579456691\\
};
\fi
\addplot [color=black!20!red,solid,mark=asterisk,mark options={solid},forget plot, every mark/.append style={mark size=2pt}]
  table[row sep=crcr]{%
-5	0\\
-4.55	1.98981754708605\\
-3.50390087956822	2.04201687488729\\
-2.53639784723702	1.63015841108467\\
-1.78716095096168	1.19608024569572\\
-1.24289341463284	0.846875899217685\\
-0.858238502552397	0.589573817977361\\
-0.590018716941354	0.40678757231779\\
-0.404423711469501	0.279157158445334\\
-0.276668849903217	0.190884464475361\\
-0.18906766383858	0.13021792594321\\
-0.129167373444976	0.0887589740129204\\
-0.0882636944772924	0.060511843436546\\
-0.0603399986978368	0.0413098108780989\\
-0.0412609005977828	0.0281659634213683\\
-0.0282401705865237	0.0191815980096613\\
-0.0193627426109386	0.0128718734383399\\
-0.0133804210171597	0.0086067152165915\\
-0.0093562033805995	0.00539905550705971\\
-0.0067774425555328	0.00377513888717923\\
-0.00496178591443204	0.0028266769753617\\
};\label{15}
\addplot [color=black!20!blue,solid,mark=asterisk,mark options={solid},dotted,forget plot, every mark/.append style={mark size=2pt}]
  table[row sep=crcr]{%
-4.5	0\\
-4.325	1.90654826555623\\
-3.32289689835451	1.9289754747395\\
-2.34835861054882	1.50339085160939\\
-1.59972717595835	1.0687265205743\\
-1.07284144754165	0.730335911412796\\
-0.71597549299476	0.490823363385567\\
-0.478096216631394	0.32804201175091\\
-0.320242540079162	0.219239858321495\\
-0.215353009444009	0.146881735587221\\
-0.145380029623626	0.0987333777018179\\
-0.0984757517111471	0.0666169160494549\\
-0.0668833242643099	0.045102791826242\\
-0.0455156395021123	0.0306368564460334\\
-0.0310128895492756	0.0208177032717631\\
-0.0211584509415111	0.0141460806169834\\
-0.0144581210441433	0.0094971361490266\\
-0.00993928811639767	0.00636534391687666\\
-0.00689352282301658	0.00405059605382408\\
-0.0049284388256951	0.00282378342529181\\
-0.00356774908139346	0.00208625556805188\\
};\label{25}
\addplot [color=black!20!green,solid,mark=asterisk,mark options={solid},dashed,forget plot, every mark/.append style={mark size=2pt}]
  table[row sep=crcr]{%
-4	0\\
-3.5	1.37787928844126\\
-2.68512990077451	1.42607571098868\\
-1.95880154675138	1.1593265026143\\
-1.39356386780675	0.868601059553158\\
-0.976333239095486	0.626634219166267\\
-0.676770256934291	0.442464963558323\\
-0.465505440317446	0.308086318492447\\
-0.318413607916023	0.212418090726944\\
-0.216972718558321	0.145428414577904\\
-0.147497098445224	0.0990822853597652\\
-0.100138660302137	0.0673390031374604\\
-0.0679372660346201	0.0457177536189703\\
-0.0460685007837663	0.0310521600902721\\
-0.0312177750399528	0.0210365819091242\\
-0.0211547576144298	0.0142243584359483\\
-0.0143436630179381	0.00945454505628462\\
-0.00978176366355571	0.00627326508640297\\
-0.00671476880763641	0.00396037187292889\\
-0.00468075504087751	0.00277616936838225\\
-0.00323114541014984	0.00204100295697578\\
};\label{35}
\end{axis}
\end{tikzpicture}%